\documentclass[11pt,a4paper,reqno]{amsart}

\usepackage{amsmath,amsthm,amssymb}
\usepackage{booktabs,array,longtable,geometry}
\usepackage[utf8]{inputenc}
\usepackage[T1]{fontenc}
\usepackage{hyperref,xcolor,enumitem,tikz,graphicx,placeins}
\hypersetup{colorlinks=true,linkcolor=black,urlcolor=black,citecolor=black}

\newtheorem{proposition}{Proposition}[section]
\newtheorem{theorem}[proposition]{Theorem}
\newtheorem{corollary}[proposition]{Corollary}

\theoremstyle{definition}
\newtheorem{definition}[proposition]{Definition}
\newtheorem{example}[proposition]{Example}
\theoremstyle{remark}
\newtheorem{remark}[proposition]{Remark}

\newcommand{\field}{\mathrm{field}}
\newcommand{\prop}[1]{\textbf{P#1}}
\newcommand{\comp}[1]{\overline{#1}}
\newcommand{\Id}{1'}
\newcommand{\Univ}{1}
\newcommand{\pabbr}[1]{\textsf{#1}}
\newcommand{\catentails}{\Rightarrow}
\newcommand{\catequiv}{\Longleftrightarrow}
\newcommand{\catnotentails}{\not\Rightarrow}

\title[A Catalogue of Properties of Binary Relations]{A Catalogue of Properties of Binary Relations: Entailments, Incompatibilities, and Independence Results}
\author{Magnus Boman}
\thanks{Affiliation: Karolinska Institutet, Stockholm, Sweden.}
\subjclass[2020]{03B45, 03G15, 06A06, 68V20}
\keywords{binary relations, entailment, relation algebra, modal correspondence, Euclideanness, anti-transitivity, independence, formal verification}

\begin{document}

\begin{abstract}
We study fifteen properties of binary relations that have established uses in
modal logic, order and preference theory, and relation algebra.  The organising
question is pragmatic: once some properties of a relation are known, which
further properties follow, which combinations force degeneracy, and which
properties remain independent?  We give a proved Horn basis of elementary and
compound entailments, explicit countermodels for non-entailments, and a
relation-algebraic translation of all fifteen properties.  The modal discussion
includes the usual Scott--Lemmon correspondences and the logic of transitive
dense frames studied by Ghilardi and Mints.  Exhaustive finite enumeration and
targeted model search are used for discovery, while the catalogue is
certified in Lean.  A kernel-checked coverage calculation ranges over all
$2^{15}=32{,}768$ antecedent sets and all fifteen possible consequents.  Its
51-rule Horn manifest is coupled definitionally to Lean proofs of semantic
soundness, and its 49 witness rows are coupled definitionally to Lean-certified
finite or infinite relations.  Consequently, over non-empty domains, the Horn
closure is complete for positive entailment among the fifteen selected
properties, and the consistency classification is complete for their positive
combinations.
\end{abstract}

\maketitle

\section{Introduction}

Binary relations are routinely described by collections of structural
properties: reflexivity, symmetry, transitivity, connexity, Euclideanness,
density and convergence.  In applications one rarely starts with all
of these properties at once.  Rather, one knows that a relation has some
properties and wants to know what follows automatically, what still needs to
be checked, and which combinations are impossible or degenerate.  The purpose
of this paper is to organise that information for a deliberately selected
family of fifteen properties.

The selection is \emph{pragmatic}, not syntactic.  We do not attempt to include
every first-order property that one can construct from a single binary
relation symbol.  Instead we collect properties with an established role in
modal frame theory, order and preference theory, or the algebraic theory of
relations; see Lemmon and Scott~\cite{LemmonScott1977}, Segerberg~\cite{Segerberg1971},
Chellas~\cite{Chellas1980}, Blackburn, de Rijke and Venema~\cite{BlackburnEtAl2001},
Fra\"iss\'e~\cite{Fraisse}, Fishburn~\cite[Chs.~1--2]{Fishburn1970},
Roberts~\cite[Ch.~2]{Roberts1979}, Tarski~\cite{Tarski1941},
Maddux~\cite{Maddux1978,Maddux}, and Schmidt~\cite{Schmidt2010}.  This criterion is meant
to make the catalogue useful as a look-up device rather than merely exhaustive
in a formal language.

The boundary is necessarily selective.  Classical preference theory contains
important further relational conditions that are not among the fifteen.  For
example, after introducing many of the properties used here, Fishburn
\cite[Sec.~2.4]{Fishburn1970} adds conditions used in the theory of interval
orders and semiorders.  Those specialised conditions are natural candidates
for a larger preference-theoretic extension of the catalogue, but including
them here would change the deliberately cross-disciplinary scope of the
present selection.

There is important prior work with a closely related encyclopaedic aim.
Burghardt~\cite{Burghardt2018} studies a broader collection of relational
properties computationally, enumerating finite relations to generate candidate
laws and then proving the surviving laws for arbitrary relations.  His later
work~\cite{Burghardt2021} develops an algebra of lifted relation properties and
combines counterexample generation with automated theorem proving.  Burghardt
therefore studies a broader collection of relational properties computationally
and algebraically; the present work isolates fifteen properties prominent in
modal, order-theoretic and relation-algebraic contexts, determines their
entailment structure under explicitly stated assumptions, and integrates this
structure with named modal frame classes.  The contribution is consequently
not that elementary implications such as ``irreflexive and transitive implies
asymmetric'' are individually new, but that a small, practically motivated
fragment is organised and checked in a structured way.

The catalogue is organised around a separation between proof, model search,
and global coverage.  Every positive entailment below has an arbitrary-domain
proof and a Lean formalisation.  Finite enumeration and Z3 are used for
discovery and countermodel search, not as substitutes for proof; finite solver
models are independently re-evaluated and certified in Lean, while infinite
witnesses are formalised directly.  The global Lean certificate uses the exact
numerical rule and witness manifests of the Horn calculation, couples them to
semantic proofs, and checks all $32{,}768$ antecedent cases by kernel
reduction.  The Python enumeration and Z3 search therefore lie outside the
trust boundary of the completeness claim.  Section~\ref{sec:verification}
gives the precise certificate structure.

Throughout, \emph{non-transitive} means simply ``not transitive'', whereas
\emph{anti-transitive} denotes the universal condition in \prop{8}.  We follow
Burghardt~\cite{Burghardt2018} in preferring \emph{anti-transitive} to the
potentially ambiguous \emph{intransitive}.  We use \emph{Euclideanness} for \prop{13}, the property traditionally
called \emph{right Euclidean}; the left Euclidean variant is not included in the
catalogue.

\section{The fifteen properties}

Throughout, $D$ is a non-empty set and $R\subseteq D\times D$.  Quantifiers
range over $D$, and $xRy$ abbreviates $(x,y)\in R$.  The relation itself is
allowed to be empty.  The \emph{field} $\field(R)$ is the set of points that
occur in at least one ordered pair of $R$.

\begin{definition}[The fifteen properties]\label{def:properties}
\leavevmode
\begin{enumerate}[label=\textbf{P\arabic*},leftmargin=2.5em,noitemsep]
  \item \textbf{Reflexivity} (\pabbr{Refl}): $\forall x\;Rxx$.
  \item \textbf{Irreflexivity} (\pabbr{Irrefl}): $\forall x\;\neg Rxx$.
  \item \textbf{Seriality} (\pabbr{Ser}): $\forall x\;\exists y\;Rxy$.
  \item \textbf{Symmetry} (\pabbr{Sym}): $\forall x,y\,(Rxy\to Ryx)$.
  \item \textbf{Antisymmetry} (\pabbr{AntiSym}): $\forall x,y\,(Rxy\wedge Ryx\to x=y)$.
  \item \textbf{Asymmetry} (\pabbr{Asym}): $\forall x,y\,(Rxy\to\neg Ryx)$.
  \item \textbf{Transitivity} (\pabbr{Trans}): $\forall x,y,z\,(Rxy\wedge Ryz\to Rxz)$.
  \item \textbf{Anti-transitivity} (\pabbr{AntiTrans}): $\forall x,y,z\,(Rxy\wedge Ryz\to\neg Rxz)$.
  \item \textbf{Negative transitivity} (\pabbr{NegTrans}): $\forall x,y,z\,(\neg Rxy\wedge\neg Ryz\to\neg Rxz)$.
  \item \textbf{Density} (\pabbr{Dense}): $\forall x,y\,(Rxy\to\exists z\,(Rxz\wedge Rzy))$.
  \item \textbf{Strong connexity} (\pabbr{SConn}): $\forall x,y\,(Rxy\vee Ryx)$.
  \item \textbf{Connexity} (\pabbr{Conn}): $\forall x,y\,(Rxy\vee Ryx\vee x=y)$.
  \item \textbf{Euclideanness} (\pabbr{Eucl}): $\forall x,y,z\,(Rxy\wedge Rxz\to Ryz)$.
  \item \textbf{Convergence} (\pabbr{Conv}): $\forall x,y,z\,(Rxy\wedge Rxz\to\exists w\,(Ryw\wedge Rzw))$.
  \item \textbf{Weak connexity} (\pabbr{WConn}): $\forall x,y,z\,(Rxy\wedge Rxz\to Ryz\vee Rzy\vee y=z)$.
\end{enumerate}
\end{definition}

The mnemonic labels are reader aids; the numbered labels \prop{1}--\prop{15}
remain the stable identifiers used throughout the formal catalogue and the
Lean development.

\begin{remark}[Terminology]
The labels for \prop{11} and \prop{12} are not standardised across logical and
preference-theoretic traditions.  Fishburn~\cite[Ch.~1]{Fishburn1970} calls
our \prop{11} \emph{connected} or \emph{complete} and our \prop{12}
\emph{weakly connected}.  Burghardt~\cite{Burghardt2018} calls our \prop{11}
\emph{connex} and our \prop{12} \emph{semi-connex}; other authors use
\emph{total} or \emph{weakly total} in related senses.  There is no single
terminological convention spanning these literatures.  Our convention
reserves \emph{strong connexity} for comparability of every ordered pair,
including the diagonal, and \emph{connexity} for comparability of distinct
points.  Thus strong connexity itself forces reflexivity.  To avoid a further
collision with the preference-theoretic use of \emph{complete}, we call
$D\times D$ the \emph{universal relation} below.
\end{remark}

\begin{remark}[Density versus order-density]\label{rem:density}
Property \prop{10} is the standard relational density condition $R\subseteq
R;R$: every edge factors into two $R$-edges, where $;$ is the relative product
used in Section~\ref{sec:algebra}.  It does \emph{not} require
the interpolating point to be distinct from the endpoints.  Consequently every
reflexive relation is dense (Proposition~\ref{prop:solo}), whereas the usual
notion of order-density for a strict order additionally excludes the endpoints.
Keeping these notions separate prevents several otherwise tempting but false
claims about finite dense relations.
\end{remark}

\section{A relation-algebraic translation}\label{sec:algebra}

The same fifteen conditions can be written compactly in the algebra of binary
relations.  This provides a second language for the catalogue and often exposes
the reason an entailment holds.  We write $\Id$ for the identity relation,
$\Univ$ for the universal relation, $R^\smile$ for converse, $\comp R$ for
complement relative to $D\times D$, juxtaposition with $\cap$ or $\cup$ for the
Boolean operations, and $R;S$ for relative product:
\[
 x(R;S)z \quad\Longleftrightarrow\quad \exists y\,(xRy\wedge ySz).
\]
These are the standard operations of relation algebra; see Tarski
\cite{Tarski1941}, Maddux~\cite{Maddux1978,Maddux}, Hirsch and
Hodkinson~\cite{HirschHodkinson2002}, and the relational treatments of Schmidt
and Str\"ohlein~\cite{SchmidtStroehlein1993} and Schmidt~\cite{Schmidt2010}.
Maddux's 1978 doctoral thesis was supervised by Tarski; his later monograph
provides a substantially expanded systematic treatment.

\begin{center}
\begin{tabular}{@{}cl@{}}
\toprule
Property & Relation-algebraic form \\
\midrule
\prop{1} & $\Id\subseteq R$ \\
\prop{2} & $R\cap\Id=\varnothing$ \\
\prop{3} & $\Id\subseteq R;R^\smile$ \\
\prop{4} & $R=R^\smile$ \\
\prop{5} & $R\cap R^\smile\subseteq\Id$ \\
\prop{6} & $R\cap R^\smile=\varnothing$ \\
\prop{7} & $R;R\subseteq R$ \\
\prop{8} & $R;R\subseteq\comp R$ \\
\prop{9} & $\comp R;\comp R\subseteq\comp R$ \\
\prop{10} & $R\subseteq R;R$ \\
\prop{11} & $\Univ=R\cup R^\smile$ \\
\prop{12} & $\comp\Id\subseteq R\cup R^\smile$ \\
\prop{13} & $R^\smile;R\subseteq R$ \\
\prop{14} & $R^\smile;R\subseteq R;R^\smile$ \\
\prop{15} & $R^\smile;R\subseteq R\cup R^\smile\cup\Id$ \\
\bottomrule
\end{tabular}
\end{center}

The translation is not merely cosmetic.  For example, \prop{1} gives
$\Id\subseteq R$, hence by monotonicity of relative product
$R=\Id;R\subseteq R;R$, which is exactly \prop{10}.  Likewise, if $R=R^\smile$
then $R^\smile;R=R;R^\smile$, so symmetry immediately entails convergence.
We nevertheless give pointwise proofs below because they require no prior
knowledge of relation algebra.

\section{Preliminary observations}

\begin{proposition}[The empty relation]\label{prop:empty}
On a non-empty domain the empty relation satisfies \prop{2}, \prop{4}, \prop{5},
\prop{6}, \prop{7}, \prop{8}, \prop{9}, \prop{10}, \prop{13}, \prop{14}, and
\prop{15}.  It fails \prop{1}, \prop{3}, and \prop{11}; it also fails \prop{12} when $|D|\ge2$, while the empty relation on a singleton satisfies \prop{12}.
\end{proposition}
\begin{proof}
Every displayed implication in \prop{2}, \prop{4}--\prop{10}, and
\prop{13}--\prop{15} is vacuously true when $R=\varnothing$; in particular,
density is vacuous because there is no pair $xRy$ to interpolate.  Reflexivity
and seriality fail because $D$ is non-empty, while strong connexity and
connexity fail on any two distinct points and, on a singleton, strong connexity
still requires the missing loop.  For connexity on a singleton the identity
disjunct makes \prop{12} true; thus the last assertion concerning \prop{12}
requires $|D|\ge2$.
\end{proof}

\begin{remark}
The singleton exception in Proposition~\ref{prop:empty} is a useful reminder
that cardinality assumptions should be stated explicitly.  In the remainder,
whenever a countermodel uses the empty relation to refute connexity, the domain
has at least two elements.
\end{remark}

Throughout the catalogue, a set of properties is called \emph{consistent}
precisely when it is jointly satisfiable by some relation on a non-empty
domain; no separate syntactic proof relation $\vdash$ is defined.  Accordingly,
\emph{inconsistent} and \emph{unsatisfiable} are used interchangeably.

\begin{proposition}[Monotonicity of entailment]\label{prop:monotonicity}
If $\Gamma\models Q$, then $\Gamma\cup\Sigma\models Q$ for every additional set
of assumptions $\Sigma$.
\end{proposition}
\begin{proof}
Every model of $\Gamma\cup\Sigma$ is, in particular, a model of $\Gamma$.
\end{proof}

\begin{remark}[Entailment notation]\label{rem:entailment-notation}
Semantic entailment is written $\Gamma\models Q$.  In the compact catalogue
displays below, $A\catentails B$ is typographical shorthand for
$A\models B$, and $A\catequiv B$ means entailment in both directions.  The
symbol $+$ joins properties conjunctively: thus
$\prop{11}\catentails\prop{1}+\prop{12}$ abbreviates the two consequences
$\prop{11}\models\prop{1}$ and $\prop{11}\models\prop{12}$.  We reserve
$\to$ for implication inside the first-order formulas defining the
properties.  Failure of catalogue entailment is written $A\catnotentails B$.
\end{remark}

For convenient reading, the mnemonic key from Definition~\ref{def:properties}
is repeated here:
\begin{center}
\small
\begin{tabular}{@{}lllll@{}}
\prop{1} \pabbr{Refl} & \prop{2} \pabbr{Irrefl} & \prop{3} \pabbr{Ser} &
\prop{4} \pabbr{Sym} & \prop{5} \pabbr{AntiSym} \\
\prop{6} \pabbr{Asym} & \prop{7} \pabbr{Trans} & \prop{8} \pabbr{AntiTrans} &
\prop{9} \pabbr{NegTrans} & \prop{10} \pabbr{Dense} \\
\prop{11} \pabbr{SConn} & \prop{12} \pabbr{Conn} & \prop{13} \pabbr{Eucl} &
\prop{14} \pabbr{Conv} & \prop{15} \pabbr{WConn}
\end{tabular}
\end{center}

\section{Entailments}\label{sec:entailments}

\subsection{Solo entailments}

\begin{proposition}[The basic solo entailments]\label{prop:solo}
For arbitrary $R$ on non-empty $D$:
\begin{enumerate}[label=(\roman*),noitemsep]
\item \prop{1} (\pabbr{Refl}) $\catentails$
      \prop{3} (\pabbr{Ser}) + \prop{10} (\pabbr{Dense});
\item \prop{4} (\pabbr{Sym}) $\catentails$ \prop{14} (\pabbr{Conv});
\item \prop{6} (\pabbr{Asym}) $\catentails$
      \prop{2} (\pabbr{Irrefl}) + \prop{5} (\pabbr{AntiSym});
\item \prop{8} (\pabbr{AntiTrans}) $\catentails$ \prop{2} (\pabbr{Irrefl});
\item \prop{11} (\pabbr{SConn}) $\catentails$
      \prop{1} (\pabbr{Refl}) + \prop{12} (\pabbr{Conn});
\item \prop{12} (\pabbr{Conn}) $\catentails$ \prop{15} (\pabbr{WConn});
\item \prop{13} (\pabbr{Eucl}) $\catentails$
      \prop{10} (\pabbr{Dense}) + \prop{14} (\pabbr{Conv}) +
      \prop{15} (\pabbr{WConn}).
\end{enumerate}
\end{proposition}
\begin{proof}
(i) If $R$ is reflexive, every $x$ has the successor $x$, so it is serial.  If
$xRy$, choose $z=x$; then $Rxz$ is $Rxx$ and $Rzy$ is $Rxy$, proving density.

(ii) Suppose $Rxy$ and $Rxz$.  Symmetry gives $yRx$ and $zRx$, so $w=x$ is a
common successor of $y$ and $z$.

(iii) If $Rxx$, asymmetry gives $\neg Rxx$, so $R$ is irreflexive.  If
$Rxy\wedge Ryx$, asymmetry contradicts one of the conjuncts; hence the
antecedent of antisymmetry can hold only when no such distinct pair exists.

(iv) Substituting $x=y=z$ into anti-transitivity yields
$Rxx\wedge Rxx\to\neg Rxx$, hence $\neg Rxx$.

(v) Put $y=x$ in strong connexity: $Rxx\vee Rxx$, hence $Rxx$.  Connexity is
immediate because $Rxy\vee Ryx$ is stronger than $Rxy\vee Ryx\vee x=y$.

(vi) Given $Rxy\wedge Rxz$, apply connexity to $y,z$; its conclusion is exactly
$Ryz\vee Rzy\vee y=z$.

(vii) Suppose first that $Rxy$.  Euclideanness applied twice to the same
target gives $Ryy$, so $z=y$ witnesses density.  If $Rxy\wedge Rxz$, right
Euclideanness gives $Ryz$, establishing weak connexity.  For convergence,
Euclideanness gives $Rzy$ from $Rxz\wedge Rxy$ and $Ryy$ from
$Rxy\wedge Rxy$; thus $w=y$ is a common successor of $y$ and $z$.
\end{proof}

\begin{corollary}\label{cor:p11closure}
Strong connexity entails \prop{1}, \prop{3}, \prop{10}, \prop{12}, \prop{14},
and \prop{15}.
\end{corollary}
\begin{proof}
Proposition~\ref{prop:solo}(v) gives reflexivity and connexity, and then
Proposition~\ref{prop:solo}(i) and (vi) give seriality, density, and weak
connexity.  For convergence, suppose $Rxy\wedge Rxz$.  Strong connexity gives
$Ryz\vee Rzy$.  In the first case choose $w=z$ and use reflexivity at $z$; in
the second choose $w=y$ and use reflexivity at $y$.
\end{proof}

\subsection{Two-premise and compound entailments}

\begin{proposition}[Strong connexity decomposes]\label{prop:p1p12p11}
\[
\prop{11}\quad\catequiv\quad \prop{1}+\prop{12}.
\]
\end{proposition}
\begin{proof}
The forward direction is Proposition~\ref{prop:solo}(v).  Conversely, fix
$x,y$.  Connexity gives $Rxy\vee Ryx\vee x=y$.  In the last case reflexivity
gives $Rxx$, which is both $Rxy$ and $Ryx$ after substituting $y=x$.  Hence
$Rxy\vee Ryx$ in every case.
\end{proof}

\begin{proposition}[Reflexivity and negative transitivity]\label{prop:p1p9}
\prop{1}+\prop{9}$\catentails$\prop{11}.
\end{proposition}
\begin{proof}
Fix $x,y$.  If neither $Rxy$ nor $Ryx$ held, negative transitivity applied to
$(x,y,x)$ would give $\neg Rxx$, contradicting reflexivity.  Thus
$Rxy\vee Ryx$.
\end{proof}

\begin{proposition}[Irreflexivity clusters]\label{prop:irr-clusters}
\begin{enumerate}[label=(\roman*),noitemsep]
\item \prop{2}+\prop{5}$\catentails$\prop{6};
\item \prop{2}+\prop{7}$\catentails$\prop{6}.
\end{enumerate}
\end{proposition}
\begin{proof}
(i) Suppose $Rxy$.  If also $Ryx$, antisymmetry gives $x=y$, contradicting
irreflexivity.  Hence $\neg Ryx$.

(ii) Suppose $Rxy$ and, towards a contradiction, $Ryx$.  Transitivity gives
$Rxx$, contradicting irreflexivity.
\end{proof}

\begin{corollary}
Every strict partial order (\prop{2}+\prop{7}) is asymmetric and
antisymmetric.
\end{corollary}

\begin{proposition}[Symmetry and Euclideanness]\label{prop:sym-eucl}
In the presence of symmetry, transitivity and Euclideanness are
equivalent:
\[
\prop{4}+\prop{7}\quad\catequiv\quad\prop{4}+\prop{13}.
\]
\end{proposition}
\begin{proof}
Assume symmetry and transitivity, and suppose $Rxy\wedge Rxz$.  Symmetry gives
$Ryx$, and transitivity applied to $Ryx\wedge Rxz$ gives $Ryz$.

Conversely, assume symmetry and Euclideanness, and suppose
$Rxy\wedge Ryz$.  Symmetry gives $Ryx$.  Applying Euclideanness at base
$y$ to $Ryx\wedge Ryz$ gives $Rxz$.
\end{proof}

\begin{proposition}[Reflexive Euclideanness]\label{prop:refl-eucl}
\prop{1}+\prop{13}$\catentails$\prop{4}+\prop{7}.
\end{proposition}
\begin{proof}
Suppose $Rxy$.  Reflexivity gives $Rxx$, and Euclideanness at base $x$
with targets $y,x$ gives $Ryx$.  Thus $R$ is symmetric.  Transitivity now
follows from Proposition~\ref{prop:sym-eucl}.
\end{proof}

\begin{theorem}[Equivalence relation / S5 cluster]\label{thm:s5}
The following are equivalent:
\begin{enumerate}[label=(\roman*),noitemsep]
\item \prop{1}+\prop{13};
\item \prop{1}+\prop{4}+\prop{7};
\item $R$ is an equivalence relation.
\end{enumerate}
\end{theorem}
\begin{proof}
(i)$\Rightarrow$(ii) is Proposition~\ref{prop:refl-eucl}; (ii)$\Rightarrow$(iii) is the
definition of an equivalence relation.  For (iii)$\Rightarrow$(i), if $Rxy\wedge Rxz$,
symmetry gives $Ryx$ and transitivity gives $Ryz$.
\end{proof}

\begin{proposition}[Symmetry plus antisymmetry]\label{prop:sym-antisym}
\prop{4}+\prop{5}$\catentails$\prop{7}.  Consequently, by
Proposition~\ref{prop:sym-eucl}, it also entails \prop{13}.
\end{proposition}
\begin{proof}
Suppose $Rxy\wedge Ryz$.  Symmetry gives $Ryx$, so antisymmetry yields $x=y$.
The premise $Ryz$ is therefore exactly $Rxz$.
\end{proof}

\begin{proposition}[Connexity and negative transitivity]\label{prop:trans-connex}
\prop{7}+\prop{12}$\catentails$\prop{9}.
\end{proposition}
\begin{proof}
Assume $\neg Rxy$ and $\neg Ryz$.  Connexity applied to $x,y$ gives
$Ryx\vee x=y$, and applied to $y,z$ gives $Rzy\vee y=z$.
If $x=y$, then $\neg Rxz$ is the assumption $\neg Ryz$; if $y=z$, it is
$\neg Rxy$.  Otherwise $Ryx$ and $Rzy$.  Suppose for contradiction that
$Rxz$.  Transitivity of $Rxz\wedge Rzy$ gives $Rxy$, contradicting the first
assumption.  Hence $\neg Rxz$.
\end{proof}

\begin{proposition}[Symmetric connex relations]\label{prop:sym-connex}
\prop{4}+\prop{12}$\catentails$\prop{9}.  If moreover $|D|\ge2$, then
\prop{4}+\prop{12}$\catentails$\prop{3}.
\end{proposition}
\begin{proof}
Symmetry and connexity imply that every two distinct points are related in both
directions.  Thus, if $\neg Rxy$, necessarily $x=y$ and the diagonal edge is
absent.  If both $\neg Rxy$ and $\neg Ryz$, we have $x=y=z$, so the conclusion
$\neg Rxz$ is one of the assumptions.  This proves negative transitivity.

For seriality, fix $x$ and choose $y\ne x$.  Connexity gives $Rxy\vee Ryx$,
and symmetry turns either disjunct into $Rxy$.
\end{proof}

\begin{proposition}[Antisymmetry with negative transitivity]\label{prop:p5p9}
\prop{5}+\prop{9}$\catentails$\prop{7}.
\end{proposition}
\begin{proof}
Suppose $Rxy\wedge Ryz$ and assume $\neg Rxz$.  The contrapositive of negative
transitivity, applied to $(y,x,z)$, says
\[
Ryz\to Ryx\vee Rxz.
\]
Hence $Ryx$.  Together with $Rxy$, antisymmetry gives $x=y$, and then $Ryz$
is $Rxz$, a contradiction.
\end{proof}

\begin{corollary}[Strict weak orders]\label{cor:strictweak}
\prop{6}+\prop{9}$\catentails$\prop{7}.  Thus an asymmetric negatively transitive
relation is a strict weak order in the usual preference-theoretic sense; see,
for example, Roberts~\cite[Ch.~2]{Roberts1979}.
\end{corollary}
\begin{proof}
Asymmetry entails antisymmetry by Proposition~\ref{prop:solo}(iii), so apply
Proposition~\ref{prop:p5p9}.
\end{proof}

\begin{proposition}[Antisymmetry with Euclideanness]\label{prop:p5p13}
\prop{5}+\prop{13}$\catentails$\prop{7}.
\end{proposition}
\begin{proof}
Suppose $Rxy\wedge Ryz$.  From $Rxy\wedge Rxy$, Euclideanness gives
$Ryy$.  Applying Euclideanness at base $y$ to $Ryz\wedge Ryy$ gives
$Rzy$.  Antisymmetry applied to $Ryz\wedge Rzy$ yields $y=z$, so $Rxy$ is
$Rxz$.
\end{proof}

\begin{proposition}[Negative transitivity with Euclideanness]\label{prop:p9p13}
\prop{9}+\prop{13}$\catentails$\prop{7}.
\end{proposition}
\begin{proof}
Suppose $Rxy\wedge Ryz$ and assume for contradiction that $\neg Rxz$.
Negative transitivity, contraposed at $(y,x,z)$, gives
$Ryz\to Ryx\vee Rxz$; hence $Ryx$.  Euclideanness at base $y$ applied to
$Ryx\wedge Ryz$ now gives $Rxz$, contradiction.
\end{proof}

\begin{proposition}[Connex Euclidean relations]\label{prop:p12p13}
\prop{12}+\prop{13}$\catentails$\prop{7}+\prop{9}.
\end{proposition}
\begin{proof}
It is enough to prove transitivity; negative transitivity then follows from
Proposition~\ref{prop:trans-connex}.  Suppose $Rxy\wedge Ryz$.  From $Rxy$,
Euclideanness gives $Ryy$.  Connexity applied to $x,z$ gives
$Rxz\vee Rzx\vee x=z$.

If $Rxz$, we are done.  If $x=z$, then $yRx$ and Euclideanness applied to
$yRx\wedge yRx$ gives $Rxx=Rxz$.  Finally suppose $Rzx$.  From
$Ryy\wedge Ryz$, Euclideanness gives $Rzy$.  From $Rzy\wedge Rzx$ it gives
$Ryx$, and from $Ryx\wedge Ryz$ it gives $Rxz$.
\end{proof}

\begin{proposition}[Reflexivity and weak connexity]\label{prop:p1p15}
\prop{1}+\prop{15}$\catentails$\prop{14}.
\end{proposition}
\begin{proof}
Suppose $Rxy\wedge Rxz$.  Weak connexity gives $Ryz$, $Rzy$, or $y=z$.
If $Ryz$, choose $w=z$ and use reflexivity of $z$; if $Rzy$, choose $w=y$ and
use reflexivity of $y$; if $y=z$, choose $w=y$ and use reflexivity once.
\end{proof}

\begin{proposition}[A three-premise recovery of Euclideanness]\label{prop:p1p4p15}
\prop{1}+\prop{4}+\prop{15}$\catentails$\prop{13}.
\end{proposition}
\begin{proof}
Suppose $Rxy\wedge Rxz$.  Weak connexity gives $Ryz$, $Rzy$, or $y=z$.
The first case is the desired conclusion; the second gives $Ryz$ by symmetry;
in the third, $Ryz$ is $Ryy$, which follows from reflexivity.
\end{proof}

\begin{proposition}[Universality from symmetry and strong connexity]\label{prop:complete}
\prop{4}+\prop{11} forces the universal relation $R=D\times D$.
\end{proposition}
\begin{proof}
Given $x,y$, strong connexity gives $Rxy\vee Ryx$.  In the second case symmetry
gives $Rxy$ as well.
\end{proof}

\begin{proposition}[Reflexivity on the field]\label{prop:fieldrefl}
If $R$ is symmetric and transitive, then it is reflexive on $\field(R)$.
\end{proposition}
\begin{proof}
Let $x\in\field(R)$.  Either $Rxy$ for some $y$ or $Ryx$ for some $y$.  By
symmetry, in either case both $Rxy$ and $Ryx$ hold.  Transitivity then gives
$Rxx$.
\end{proof}

\begin{proposition}[Further compound entailments from closure reconciliation]
\label{prop:closure-entailments}
For arbitrary $R$ on a non-empty domain:
\begin{enumerate}[label=(\roman*),noitemsep]
\item \prop{3}+\prop{4}+\prop{7}$\catentails$\prop{1};
\item \prop{3}+\prop{8}+\prop{9}$\catentails$\prop{4};
\item \prop{8}+\prop{9}+\prop{14}$\catentails$\prop{4};
\item \prop{3}+\prop{7}+\prop{9}$\catentails$\prop{14};
\item \prop{3}+\prop{7}+\prop{15}$\catentails$\prop{14};
\item \prop{3}+\prop{8}+\prop{15}$\catentails$\prop{14};
\item \prop{3}+\prop{4}+\prop{9}+\prop{15}$\catentails$\prop{12}.
\end{enumerate}
\end{proposition}
\begin{proof}
(i) Fix $x$.  By seriality choose $y$ with $Rxy$.  Symmetry gives $Ryx$, and
transitivity yields $Rxx$.

(ii) Suppose $Rxy$.  By seriality choose $z$ with $Ryz$.  Anti-transitivity
gives $\neg Rxz$.  If $\neg Ryx$, then negative transitivity applied to
$(y,x,z)$ gives $\neg Ryz$, a contradiction.  Hence $Ryx$.

(iii) Suppose $Rxy$.  Applying convergence to the repeated pair
$Rxy\wedge Rxy$ gives some $w$ with $Ryw$.  Anti-transitivity applied to
$Rxy\wedge Ryw$ gives $\neg Rxw$.  If $\neg Ryx$, negative transitivity at
$(y,x,w)$ gives $\neg Ryw$, a contradiction.  Hence $Ryx$.

(iv) Suppose $Rxy\wedge Rxz$.  If $Ryz$, choose by seriality some $w$ with
$Rzw$; transitivity gives $Ryw$, so $w$ is a common successor.  If
$\neg Ryz$, choose by seriality some $w$ with $Ryw$.  Were $\neg Rzw$ also
true, negative transitivity at $(y,z,w)$ would contradict $Ryw$.  Thus $Rzw$,
and again $w$ is common.

(v) Suppose $Rxy\wedge Rxz$.  Weak connexity gives $Ryz$, $Rzy$, or $y=z$.
In the first case choose a successor $w$ of $z$ and use transitivity to obtain
$Ryw$; in the second choose a successor $w$ of $y$ and use transitivity to
obtain $Rzw$; in the equality case any successor of $y=z$ is common.

(vi) Suppose $Rxy\wedge Rxz$.  Weak connexity again gives $Ryz$, $Rzy$, or
$y=z$.  The first alternative contradicts anti-transitivity applied to
$Rxy\wedge Ryz$, and the second contradicts anti-transitivity applied to
$Rxz\wedge Rzy$.  Hence $y=z$.  Seriality supplies a common successor.

(vii) Fix $x,y$.  If $Rxy$, connexity is immediate.  Otherwise choose, by
seriality, $z$ with $Rxz$.  Negative transitivity shows that $Ryz$: if
$\neg Ryz$, then $\neg Rxy\wedge\neg Ryz$ would imply $\neg Rxz$.
By symmetry $zRx$ and $zRy$.  Weak connexity at base $z$ gives
$Rxy\vee Ryx\vee x=y$, which is connexity.
\end{proof}

\section{Incompatibility and forced degeneracy}\label{sec:incompat}

It is useful to distinguish three phenomena: assumptions that are genuinely
inconsistent on a non-empty domain, assumptions that force $R$ to be empty,
and assumptions that merely eliminate composable pairs.

\begin{proposition}[Unsatisfiable combinations]\label{prop:unsat}
On a non-empty domain the following combinations are unsatisfiable:
\[
\prop{1}+\prop{2},\qquad
\prop{1}+\prop{6},\qquad
\prop{1}+\prop{8},\qquad
\prop{2}+\prop{11}.
\]
The combinations \prop{6}+\prop{11} and \prop{8}+\prop{11} are therefore
unsatisfiable as inherited cases.
\end{proposition}
\begin{proof}
The first pair is immediate.  The second and third follow because asymmetry and
anti-transitivity each entail irreflexivity.  The fourth follows because strong
connexity entails reflexivity.  The inherited cases use the same solo
entailments.
\end{proof}

\begin{proposition}[Pairs forcing the empty relation]\label{prop:forceempty}
Each of the following combinations forces $R=\varnothing$:
\[
\prop{4}+\prop{6},\qquad
\prop{2}+\prop{13},\qquad
\prop{8}+\prop{10}.
\]
Consequently \prop{6}+\prop{13} and \prop{8}+\prop{13} also force
$R=\varnothing$.
\end{proposition}
\begin{proof}
If $Rxy$ and $R$ is symmetric, then $Ryx$, contradicting asymmetry; hence
\prop{4}+\prop{6} forces emptiness.

If $Rxy$ and $R$ is Euclidean, applying \prop{13} to
$Rxy\wedge Rxy$ gives $Ryy$, contradicting irreflexivity.  Thus
\prop{2}+\prop{13} forces emptiness.

Finally, if $Rxy$ and $R$ is dense, choose $z$ with $Rxz\wedge Rzy$.
Anti-transitivity applied to this composable pair gives $\neg Rxy$, a
contradiction.  Hence \prop{8}+\prop{10} forces emptiness.  The two stated
corollaries use \prop{6}$\catentails$\prop{2},
\prop{8}$\catentails$\prop{2}, and \prop{13}$\catentails$\prop{10}.
\end{proof}

\begin{proposition}[Transitivity plus anti-transitivity]\label{prop:p7p8}
A relation is both transitive and anti-transitive if and only if it has no
composable pair:
\[
\prop{7}+\prop{8}\quad\Longleftrightarrow\quad R;R=\varnothing.
\]
In particular, non-empty relations can satisfy both properties.
\end{proposition}
\begin{proof}
If $Rxy\wedge Ryz$, transitivity gives $Rxz$ while anti-transitivity gives
$\neg Rxz$, impossible.  Hence no composable pair exists.  Conversely, if no
composable pair exists, the antecedents of both universal implications are
always false, so both properties hold vacuously.
\end{proof}

\begin{example}
On $D=\{a,b\}$, the one-edge relation $R=\{(a,b)\}$ is non-empty and satisfies
both transitivity and anti-transitivity.  This is the smallest counterexample
to the tempting but false claim that \prop{7}+\prop{8} is incompatible for
non-empty $R$.
\end{example}

\begin{corollary}[Seriality with no composable pairs]\label{cor:p3p7p8}
The following observation was pointed out by Ian Hodkinson.  On a non-empty domain,
\[
\prop{3}+\prop{7}+\prop{8}
\]
is unsatisfiable.
\end{corollary}
\begin{proof}
By Proposition~\ref{prop:p7p8}, \prop{7}+\prop{8} leaves no composable
pair.  But seriality gives, for any $x$, some $y$ with $Rxy$ and then some
$z$ with $Ryz$, producing a composable pair.
\end{proof}

\begin{proposition}[Convergence with no composable pairs]
\label{prop:nocomp-conv}
If $R;R=\varnothing$ and $R$ is convergent, then $R=\varnothing$.
Consequently \prop{7}+\prop{8}+\prop{14} forces the empty relation.
\end{proposition}
\begin{proof}
Suppose $Rxy$.  Apply convergence to the repeated premise
$Rxy\wedge Rxy$.  There is then some $w$ with $Ryw$, so $Rxy\wedge Ryw$ is a
composable pair, contradicting $R;R=\varnothing$.  Hence no edge exists.  The
consequence follows from Proposition~\ref{prop:p7p8}.
\end{proof}

\begin{corollary}[Four characterisations of the empty relation]
\label{cor:empty-characterisations}
On a non-empty domain the following conditions are equivalent:
\[
R=\varnothing,\qquad
\prop{4}+\prop{6},\qquad
\prop{2}+\prop{13},\qquad
\prop{8}+\prop{10},\qquad
\prop{7}+\prop{8}+\prop{14}.
\]
Consequently, whenever any one of the four property combinations holds,
$R$ also has the full empty-relation profile
\[
\prop{2}+\prop{4}+\prop{5}+\prop{6}+\prop{7}+\prop{8}+\prop{9}+
\prop{10}+\prop{13}+\prop{14}+\prop{15}.
\]
\end{corollary}
\begin{proof}
The four right-to-left implications are
Proposition~\ref{prop:forceempty} and
Proposition~\ref{prop:nocomp-conv}.  Conversely the empty relation has all
four displayed combinations, by Proposition~\ref{prop:empty}.  The profile
statement is again Proposition~\ref{prop:empty}.
\end{proof}

\subsection{Finite-domain strengthening}\label{subsec:finite}

The finite-model distinction in this subsection was suggested by Ian Hodkinson.
Finiteness adds genuine consequences that fail on arbitrary domains, and the
following three elementary cases explain the main strict-order effect.

\begin{proposition}[Finite strict-order degeneracy]\label{prop:finite-strict}
Let $D$ be finite and non-empty.
\begin{enumerate}[label=(\roman*),noitemsep]
\item \prop{2}+\prop{3}+\prop{7} is unsatisfiable;
\item \prop{2}+\prop{7}+\prop{10} holds if and only if $R=\varnothing$;
\item \prop{2}+\prop{7}+\prop{14} holds if and only if $R=\varnothing$.
\end{enumerate}
\end{proposition}
\begin{proof}
Under \prop{2}+\prop{7}, $R$ is a strict partial order.

(i) Every finite strict partial order has a maximal element, whereas seriality
requires every element to have a strict successor.

(ii) Suppose $R$ is non-empty.  Choose a chain of maximal finite length.
Density applied to any adjacent edge of that chain supplies a point strictly
between its endpoints, extending the chain, a contradiction.  Hence
$R=\varnothing$.  The converse follows from Proposition~\ref{prop:empty}.

(iii) Again suppose $R$ is non-empty and choose an edge $xRy$ with $y$
maximal among points lying above some predecessor.  Applying convergence to
the repeated premise $Rxy\wedge Rxy$ gives $yRw$ for some $w$, contradicting
maximality.  The converse is again Proposition~\ref{prop:empty}.
\end{proof}

\begin{remark}
These implications are genuinely finite-only.  Infinite strict orders provide
counterexamples: the Lean-certified witness library contains serial strict
orders, dense strict orders such as $(\mathbb{Q},<)$, and non-empty
irreflexive transitive convergent relations.  We do not claim here to give a
complete catalogue of all additional finite-domain consequences; a dedicated
finite-only reconciliation is a natural small extension of the verification
analysis.
\end{remark}

\section{Independence and non-entailment}\label{sec:independence}

We write $P\parallel Q$ when neither $P\models Q$ nor $Q\models P$.  In that
case two countermodels are required.  For compound antecedents, where a
symmetric formulation is usually unnatural, we use the weaker phrase
\emph{non-entailment} and give an explicit countermodel.

\begin{remark}[Logical reading of independence]\label{rem:logical-independence}
If $M_P$ denotes the class of models of a property theory $P$, then
$P\models Q$ means $M_P\subseteq M_Q$.  Non-entailment is therefore directed:
$P\not\models Q$ says that $M_P\setminus M_Q$ is non-empty.  Our notation
$P\parallel Q$ is deliberately symmetric and abbreviates the conjunction of
the two directed non-entailments.  This is the classical
\emph{independence-of-axioms} viewpoint familiar from elementary logical
treatments such as Suppes~\cite{Suppes1957} and Stoll~\cite{Stoll1963}, not
independence of primitive symbols in the sense of definability methods such as
Padoa's.

Compatibility is a separate issue: $P$ and $Q$ are jointly satisfiable exactly
when $M_P\cap M_Q\neq\varnothing$.  Thus incompatibility is not a limiting
case of independence but an orthogonal semantic distinction.  One could
introduce quantitative degrees of dependence on a fixed finite domain by
counting models, but over arbitrary domains there is no canonical measure on
these model classes, so no such degree is assumed here.
\end{remark}

\begin{proposition}[Selected independence results]\label{prop:independence}
The following pairs are independent:
\begin{enumerate}[label=(\roman*),noitemsep]
\item \prop{1}$\parallel$\prop{12};
\item \prop{3}$\parallel$\prop{7};
\item \prop{6}$\parallel$\prop{7};
\item \prop{7}$\parallel$\prop{9};
\item \prop{7}$\parallel$\prop{10};
\item \prop{7}$\parallel$\prop{13};
\item \prop{9}$\parallel$\prop{12};
\item \prop{14}$\parallel$\prop{15}.
\end{enumerate}
\end{proposition}
\begin{proof}
We give two witnesses for each pair.

(i) On $\{a,b\}$ the identity relation is reflexive but not connex.  The
one-edge relation $\{(a,b)\}$ is connex but not reflexive.

(ii) The two-cycle $\{(a,b),(b,a)\}$ is serial but not transitive.  The relation
$\{(a,a)\}$ on $\{a,b\}$ is transitive but not serial.

(iii) The directed three-cycle
$\{(a,b),(b,c),(c,a)\}$ is asymmetric but not transitive.  The identity
relation on two points is transitive but not asymmetric.

(iv) The two-cycle is negatively transitive but not transitive.  The relation
$\{(a,a)\}$ on $\{a,b\}$ is transitive but not negatively transitive: take
$x=a,y=b,z=a$.

(v) Let $D=\{a,b\}$ and
$R=\{(a,a),(a,b),(b,a)\}$.  Every edge factors through $a$, so $R$ is dense,
but $bRa\wedge aRb$ while $bRb$ fails, so $R$ is not transitive.  Conversely,
$\{(a,b)\}$ is transitive but not dense.

(vi) The one-edge relation $\{(a,b)\}$ is transitive but not Euclidean,
because $Rab\wedge Rab$ would require $Rbb$.  For the other direction, on
$D=\{a,b,c\}$ let
\[
R=\{(a,a),(a,b),(b,a),(b,b),(c,a)\}.
\]
The successor sets are $\{a,b\}$, $\{a,b\}$, and $\{a\}$ respectively, and
each is Euclidean; nevertheless $cRa\wedge aRb$ while $cRb$ fails, so
transitivity fails.

(vii) The empty relation on a two-element domain is negatively transitive but
not connex.  The directed three-cycle is connex but not negatively transitive:
$\neg aRc$ and $\neg cRb$ while $aRb$.

(viii) Let
$R=\{(a,b),(a,c),(b,a),(c,a)\}$ on three points.  It is symmetric, hence
convergent by Proposition~\ref{prop:solo}(ii), but it is not weakly connex:
$a$ has successors $b,c$, which are distinct and incomparable.  Conversely,
$\{(a,b)\}$ is weakly connex vacuously but not convergent, since the repeated
premise $Rab\wedge Rab$ would require $b$ to have a successor common with
itself.
\end{proof}

\begin{proposition}[Selected compound non-entailments]\label{prop:compound-nonent}
The following implications fail:
\begin{enumerate}[label=(\roman*),noitemsep]
\item \prop{3}+\prop{13} $\catnotentails$ \prop{7};
\item \prop{7}+\prop{14} $\catnotentails$ \prop{13};
\item \prop{7}+\prop{15} $\catnotentails$ \prop{13};
\item \prop{4}+\prop{15} $\catnotentails$ \prop{13};
\item \prop{1}+\prop{7}+\prop{15} $\catnotentails$ \prop{13};
\item \prop{2}+\prop{7}+\prop{10} $\catnotentails$ \prop{9}.
\end{enumerate}
\end{proposition}
\begin{proof}
(i) Use the three-point Euclidean non-transitive relation from
Proposition~\ref{prop:independence}(vi); it is serial.

(ii) On $\{a,b\}$ let $R=\{(a,a),(a,b),(b,b)\}$.  It is transitive and, being
reflexive, convergent by Proposition~\ref{prop:p1p15} together with its weak
connexity; directly, every pair of successors has a common successor.  It is
not Euclidean because $Raa\wedge Rab$ but $Rba$ fails.

(iii) The one-edge relation $\{(a,b)\}$ is transitive and weakly connex but not
Euclidean.

(iv) The two-cycle $\{(a,b),(b,a)\}$ is symmetric and weakly connex but not
Euclidean, again because a repeated target would require a loop.

(v) The relation $\{(a,a),(a,b),(b,b)\}$ used in (ii) is reflexive, transitive,
and weakly connex but not Euclidean.

(vi) On $\mathbb{Q}^2$ use the componentwise strict order:
$(a,b)R(c,d)$ iff $a<c$ and $b<d$.  It is irreflexive and transitive.  It is
dense, since coordinatewise midpoints interpolate every edge.  Negative
transitivity fails for
$x=(0,0)$, $y=(2,0)$, $z=(1,1)$: $\neg xRy$ and $\neg yRz$, but $xRz$.
\end{proof}

\begin{remark}[Why finite search is not a proof]
The last countermodel is necessarily infinite.  Indeed, on a finite non-empty
domain Proposition~\ref{prop:finite-strict}(ii) shows that
\prop{2}+\prop{7}+\prop{10} forces $R=\varnothing$, and
Proposition~\ref{prop:empty} then gives \prop{9}.  Hence no finite relation can
witness
\[
\prop{2}+\prop{7}+\prop{10}\catnotentails\prop{9}.
\]
This gives a concrete example of why finite search, although valuable for
countermodel discovery, cannot certify arbitrary-domain non-entailment.
\end{remark}

\section{Named structures and the corrected cluster diagram}\label{sec:clusters}

The preceding results recover several familiar classes.
\begin{description}[style=nextline]
\item[Preorder (quasi-order)]
Defined by \prop{1}+\prop{7}.  It also satisfies \prop{3} and \prop{10}.
The synonymous terms \emph{quasi-order} and \emph{quasi-ordering} have a long
logical and preference-theoretic history: Suppes~\cite{Suppes1957} uses
\emph{quasi-ordering}, while Fishburn~\cite{Fishburn1970} records both
\emph{quasi-order} and \emph{preorder}.

\item[Partial order]
Defined by \prop{1}+\prop{5}+\prop{7}.  It is a preorder with
antisymmetry.

\item[Total preorder]
Defined by \prop{1}+\prop{7}+\prop{12}.  Proposition~\ref{prop:p1p12p11}
gives strong connexity, and Proposition~\ref{prop:trans-connex} gives negative
transitivity.

\item[Total order]
Defined by \prop{1}+\prop{5}+\prop{7}+\prop{12}.  It is both a partial
order and a total preorder.

\item[Strict partial order]
Defined by \prop{2}+\prop{7}.  It is asymmetric and antisymmetric by
Proposition~\ref{prop:irr-clusters}.

\item[Strict total order]
Defined by \prop{2}+\prop{7}+\prop{12}.  It additionally satisfies
negative transitivity by Proposition~\ref{prop:trans-connex}.

\item[Strict weak order]
Defined by asymmetry and negative transitivity, \prop{6}+\prop{9};
transitivity follows from Corollary~\ref{cor:strictweak}.  This is the standard
preference-theoretic class in which incomparability is transitive; see also
Fishburn~\cite{Fishburn1970}.

\item[Equivalence relation]
Defined by \prop{1}+\prop{4}+\prop{7}, equivalently by
\prop{1}+\prop{13} (Theorem~\ref{thm:s5}).  Every equivalence relation also
satisfies \prop{3}, \prop{10}, \prop{13}, \prop{14}, and \prop{15}, but it
need not satisfy connexity or negative transitivity: the identity relation on
two points refutes both.

\item[Universal relation]
The relation $D\times D$.  Proposition~\ref{prop:complete} shows that symmetry
plus strong connexity already forces this class.
\end{description}

Figure~\ref{fig:clusters} records only genuine inclusions between named classes.
An arrow from $A$ to $B$ means that every $A$-relation is a $B$-relation, so $A$
is the more specialised class.

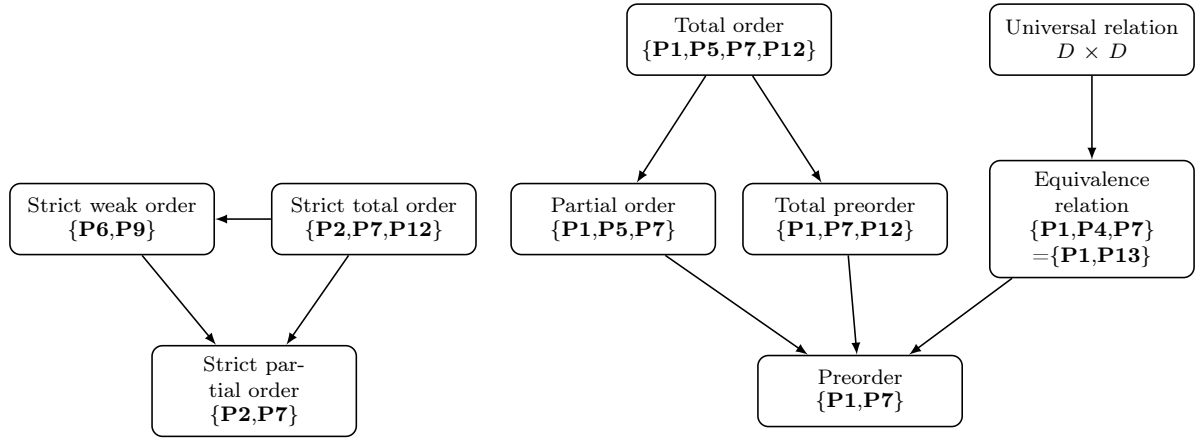
\begin{figure}[htbp]
\centering
\resizebox{0.98\textwidth}{!}{%
\begin{tikzpicture}[>=latex,semithick,
  nd/.style={draw,rounded corners=4pt,text width=2.45cm,minimum height=0.95cm,
             align=center,font=\scriptsize,inner sep=4pt}]

% strict branch
\node[nd] (spo) at (1.9,0) {Strict partial order\\[-1pt]\{\prop{2},\prop{7}\}};
\node[nd] (swo) at (0.0,2.3) {Strict weak order\\[-1pt]\{\prop{6},\prop{9}\}};
\node[nd] (sto) at (3.5,2.3) {Strict total order\\[-1pt]\{\prop{2},\prop{7},\prop{12}\}};
\draw[->] (swo) -- (spo);
\draw[->] (sto) -- (spo);
\draw[->] (sto) -- (swo);

% reflexive branch
\node[nd] (pre) at (10.0,0) {Preorder\\[-1pt]\{\prop{1},\prop{7}\}};
\node[nd] (po) at (6.7,2.3) {Partial order\\[-1pt]\{\prop{1},\prop{5},\prop{7}\}};
\node[nd] (tpre) at (9.8,2.3) {Total preorder\\[-1pt]\{\prop{1},\prop{7},\prop{12}\}};
\node[nd] (eq) at (13.1,2.3) {Equivalence relation\\[-1pt]\{\prop{1},\prop{4},\prop{7}\}\\=\{\prop{1},\prop{13}\}};
\node[nd] (to) at (8.25,4.7) {Total order\\[-1pt]\{\prop{1},\prop{5},\prop{7},\prop{12}\}};
\node[nd] (comp) at (13.1,4.7) {Universal relation\\[-1pt]$D\times D$};
\draw[->] (po) -- (pre);
\draw[->] (tpre) -- (pre);
\draw[->] (eq) -- (pre);
\draw[->] (to) -- (po);
\draw[->] (to) -- (tpre);
\draw[->] (comp) -- (eq);
\end{tikzpicture}%
}
\caption{Selected named relation classes.  Every arrow is an inclusion of
classes in the direction shown; no arrow represents a construction such as
reflexive closure.}
\label{fig:clusters}
\end{figure}
\FloatBarrier

\begin{remark}[Strict versus reflexive order terminology]
Figure~\ref{fig:clusters} deliberately contains no inclusion arrow between
strict partial orders and partial orders.  On a non-empty domain no relation can
belong to both classes: a strict partial order is irreflexive, whereas a partial
order is reflexive.  The linguistic closeness of the names reflects a standard
change of representation rather than inclusion of classes.  If $<$ is a strict
partial order, its reflexive closure $\leq\;=\;<\,\cup\Id$ is a partial order;
conversely, the strict part of a partial order $\leq$, defined by
$x<y$ iff $x\leq y$ and $x\neq y$, is a strict partial order.  The figure records
entailment between properties of a fixed relation, so these transformations are
not represented by arrows.
\end{remark}

\section{Modal frame correspondences}\label{sec:modal}

Several properties in Definition~\ref{def:properties} are standard first-order
conditions on Kripke frames.  In the usual normal modal setting, reflexivity,
seriality, symmetry, transitivity, and Euclideanness correspond to the
familiar axioms $T$, $D$, $B$, $4$, and $5$, respectively; see Lemmon and
Scott~\cite{LemmonScott1977}, Chellas~\cite{Chellas1980}, and Blackburn,
de Rijke and Venema~\cite{BlackburnEtAl2001}.  Convergence \prop{14} is the
frame condition for the familiar $\mathsf{.2}$-type principle
$\Diamond\Box A\to\Box\Diamond A$ in the reflexive-transitive setting.

The density condition \prop{10} deserves separate emphasis because its logical
position is easy to misread.  Its modal correspondent is the converse-looking
principle $\Box\Box A\to\Box A$.  Ghilardi and Mints~\cite{GhilardiMints2014}
study the logic $\mathsf{K4De}$ of transitive dense frames, obtained by
replacing the reflexivity component of the usual $\mathsf{S4}$ presentation by
a density axiom.  This is entirely consistent with Proposition~\ref{prop:solo}:
reflexivity entails relational density, so density is weaker than reflexivity,
not an additional strengthening of it.  In particular, adding \prop{10} to an
already reflexive frame imposes no new relational constraint.

Theorem~\ref{thm:s5} gives the standard S5 cluster in a compact form:
reflexivity plus Euclideanness already yields symmetry and transitivity.
Likewise, the Scott--Lemmon convergence schema explains why conditions such as
Euclideanness and convergence sit naturally together, although the
entailment relations between our fifteen predicates must still be proved at the
relational level rather than read off from the schema.

A terminological distinction is also useful for S4.2.  Dummett and Lemmon
\cite{DummettLemmon1959} use a stronger global directedness condition in their
completeness argument, whereas \prop{14} only requires a common successor for
points that already share a predecessor.  The identity relation on any set of
size at least two is reflexive, transitive, and convergent, but two distinct
points have no common successor; thus local convergence and global directedness
should not be conflated.

\section{Verification workflow and formalisation}\label{sec:verification}

The global completeness question for a catalogue is different from the proof of
any individual implication.  With fifteen predicates there are
$2^{15}=32{,}768$ possible antecedent sets, and the assertion that the
catalogue has no omitted positive consequences is itself a finite but
non-trivial coverage claim.  The verification architecture therefore separates discovery
from certification.

\subsection{Exhaustive small-model semantics}

An accompanying Python program exhaustively enumerates all binary relations on
domains of sizes $1$ through $4$ and evaluates \prop{1}--\prop{15}.  The numbers
of relations are respectively
\[
2,\qquad 16,\qquad 512,\qquad 65{,}536,
\]
for a total of $66{,}066$ relations.  These collapse to $188$ distinct
fifteen-bit property signatures.  The enumeration is used to discover
candidate entailments, small countermodels, and finite-survival cases; it is
not part of the proof of arbitrary-domain completeness.  This distinction
matters because antecedents such as \prop{2}+\prop{3}+\prop{7} have no finite
models although they are realised by infinite strict orders.

\subsection{Horn-normalised reconciliation}

The proved entailments and degeneracy rules form a Horn system.  Each
antecedent set is replaced by its Horn closure, so logically equivalent
obligations are not counted repeatedly under different literal presentations.
An antecedent with no small finite model requires either an explicit infinite
model or a proof of inconsistency.  This Horn-normalised reconciliation
organises the proof obligations; the global Lean certificate below verifies
the resulting coverage independently of the discovery procedure.

\subsection{Countermodel discovery and Lean certification}

For larger finite countermodels Z3 serves as a discovery tool.  A solver model
is not accepted on its own: its full \prop{1}--\prop{15} profile is
re-evaluated by an independent Python evaluator and encoded as an explicit
finite relation in Lean.  Each advertised property and non-property is checked
there by kernel reduction using \texttt{decide +kernel}, while infinite
witnesses are defined directly in Lean and proved from the ordered structures
involved.  The supplemental witness library includes, among others, the usual
strict orders on $\mathbb{N}$ and $\mathbb{Q}$, the componentwise strict order
on $\mathbb{Q}^2$, two disjoint rational chains, a dense fork, a strict weak
order with non-singleton indifference classes, and several solver-discovered
relations on five to nine points.  Twenty-one additional relations on at most
four points were selected from the exhaustive finite enumeration and given
complete fifteen-property profiles directly in Lean.  Thus the finite enumeration serves only as a discovery source; every model
used by the coverage certificate has an independent Lean-checked semantic
profile.

\subsection{Kernel-checked global coverage}\label{subsec:kernel-coverage}

The Lean certificate has three linked parts.  First,
\texttt{GlobalEntailmentCertificate.lean} implements an expanded 51-rule Horn
system on nineteen bits: the fifteen properties together with auxiliary bits
for the empty relation, absence of composable pairs, inconsistency, and the
universal relation.  Its witness manifest contains 49 certified rows.  Lean
kernel reduction checks both that no witness row conflicts with the Horn basis
and that, for every one of the $32{,}768$ antecedent masks, either inconsistency
is reached or a certified model exists and every property outside Horn closure
has a certified refuter.  Equivalently, all $32{,}768\times 15$ potential
positive-entailment obligations are covered.

Second, \texttt{WitnessCoupling.lean} packages every numerical witness row with
its actual carrier, relation, and proof of the stated positive and negative
properties.  The proof-carrying list projects by definitional equality to the
exact 49-row numerical manifest used by the coverage calculation.  Third,
\texttt{RuleCoupling.lean} proves semantic soundness of every one of the 51
numerical Horn rules on arbitrary non-empty domains; again, the proof-carrying
list projects by definitional equality to the exact rule list used by the
closure computation.  Hence neither the rule table nor the countermodel table
enters the global calculation through an unchecked transcription.

The consistency side is independently compressed in
\texttt{ConsistencyCertificate.lean}: a kernel computation over all
$32{,}768$ positive masks shows that each mask either contains one of the
seventeen inclusion-minimal inconsistent sets or is contained in one of the
six inclusion-maximal consistent sets.  The six maximal sets have exact
Lean-certified models, and all seventeen minimal inconsistent sets have direct
Lean contradiction proofs.

The formalisation is modular rather than packaged as a single monolithic Lean
theorem with semantic entailment notation.  Taken together, however, the
kernel-checked rule soundness, exact witness semantics, definitional coupling
to the numerical manifests, and exhaustive coverage calculation establish the
following mathematical statement.

\begin{theorem}[Completeness of the fifteen-property catalogue]
\label{thm:catalogue-complete}
Let $\Gamma\subseteq\{\prop{1},\ldots,\prop{15}\}$ and let \prop{$i$} be one
of the fifteen properties, with relations interpreted on non-empty domains.
Then:
\begin{enumerate}[label=(\roman*),noitemsep]
\item if the Horn closure of $\Gamma$ reaches $\bot$, then $\Gamma$ is
inconsistent;
\item otherwise, $\Gamma\models\prop{$i$}$ if and only if \prop{$i$} belongs
to its Horn closure;
\item $\Gamma$ is consistent if and only if it is contained in at least one of
M1--M6 in Section~\ref{sec:global-summaries}, equivalently if and only if it
contains none of the seventeen minimal inconsistent sets listed there.
\end{enumerate}
\end{theorem}

\begin{proof}
Soundness of membership in Horn closure follows by iteration of the 51
semantically proved Horn rules.  If \prop{$i$} is absent from a consistent
closure, the exhaustive coverage certificate supplies one of the semantically
coupled witnesses satisfying the closure and falsifying \prop{$i$}; hence the
entailment fails.  The inconsistency and consistency claims follow from the
independently kernel-checked six-MCS/seventeen-MIS cover together with the six
model certificates and seventeen contradiction proofs.
\end{proof}

This theorem is deliberately scoped.  It asserts completeness for positive
entailment and compatibility among the fifteen selected properties under the
standing non-empty-domain semantics.  It is not a claim of completeness for
arbitrary first-order consequences about binary relations, nor a claim that all
additional consequences peculiar to finite domains have been catalogued.

\subsection{Reproducibility materials}

The verification package accompanying the preprint and journal submission
contains the Lean project with its pinned toolchain, the Python enumeration and
reconciliation programs, the Z3 search program and solver-discovered witness
files, and a short \texttt{VERIFICATION.md} giving the exact commands for the
formal layers.  The release is preserved in the public repository
\url{https://github.com/d3dd/binary-relations-catalogue} and in the immutable
Zenodo archive
\href{https://doi.org/10.5281/zenodo.22693298}{10.5281/zenodo.22693298}.
Raw build and verification transcripts are
convenience records rather than proofs: the reproducibility claim rests on
rerunnable source, explicit toolchain versions, and Lean kernel checking of the
theorem, rule, witness, consistency, and global coverage certificates.  The
Python and Z3 components document how witnesses and candidate laws were found;
they are not trusted for Theorem~\ref{thm:catalogue-complete}.

\section{Compact global consistency summaries}\label{sec:global-summaries}

Following a suggestion by Ian Hodkinson, we record two compact global
summaries.  He independently calculated from an earlier draft the same six
maximal consistent sets and seventeen minimal inconsistent sets.  They are also certified by the Lean formalisation described in
Section~\ref{sec:verification}.  The closure calculation gives a slightly
broader count that is small enough to state directly.  There are $358$ distinct
subsets of \{\prop{1},\ldots,\prop{15}\} closed under the Horn consequence
operator.  One is the all-properties closure reached from an inconsistent
antecedent; the remaining $357$ are consistent closed positive sets.

\subsection{Maximal consistent property sets}

Exactly six subsets are inclusion-maximal among the jointly satisfiable
positive property sets.  Each row below is also the exact
\prop{1}--\prop{15} profile of the witness shown.

\begin{center}
\small
\begin{tabular}{@{}cl@{}}
\toprule
 & Maximal consistent set and a simple witness \\
\midrule
M1 &
$\{\prop{1},\prop{3},\prop{4},\prop{5},\prop{7},\prop{9},\prop{10},
\prop{11},\prop{12},\prop{13},\prop{14},\prop{15}\}$:
one reflexive point \\[2pt]
M2 &
$\{\prop{2},\prop{3},\prop{4},\prop{8},\prop{9},\prop{12},\prop{14},
\prop{15}\}$:
two irreflexive points joined symmetrically \\[2pt]
M3 &
$\{\prop{2},\prop{3},\prop{4},\prop{9},\prop{10},\prop{12},\prop{14},
\prop{15}\}$:
the complete irreflexive relation on three points \\[2pt]
M4 &
$\{\prop{2},\prop{3},\prop{5},\prop{6},\prop{7},\prop{9},\prop{10},
\prop{12},\prop{14},\prop{15}\}$:
$(\mathbb{Q},<)$ \\[2pt]
M5 &
$\{\prop{2},\prop{3},\prop{5},\prop{6},\prop{8},\prop{12},\prop{14},
\prop{15}\}$:
a directed three-cycle \\[2pt]
M6 &
$\{\prop{2},\prop{4},\prop{5},\prop{6},\prop{7},\prop{8},\prop{9},
\prop{10},\prop{12},\prop{13},\prop{14},\prop{15}\}$:
the empty relation on one point \\
\bottomrule
\end{tabular}
\end{center}

Thus a positive property set is consistent exactly when it is contained in at
least one of M1--M6.  Notice that M4 necessarily requires an infinite domain
by Proposition~\ref{prop:finite-strict}(i).

\subsection{Minimal inconsistent property sets}

Dually, there are exactly seventeen inclusion-minimal inconsistent subsets:
\[
\begin{array}{llll}
\{\prop{1},\prop{2}\} &
\{\prop{1},\prop{6}\} &
\{\prop{1},\prop{8}\} &
\{\prop{2},\prop{3},\prop{4},\prop{5}\} \\[2pt]
\{\prop{2},\prop{3},\prop{4},\prop{7}\} &
\{\prop{2},\prop{3},\prop{13}\} &
\{\prop{2},\prop{11}\} &
\{\prop{3},\prop{4},\prop{5},\prop{8}\} \\[2pt]
\{\prop{3},\prop{4},\prop{6}\} &
\{\prop{3},\prop{5},\prop{8},\prop{9}\} &
\{\prop{3},\prop{6},\prop{8},\prop{9}\} &
\{\prop{3},\prop{6},\prop{13}\} \\[2pt]
\{\prop{3},\prop{7},\prop{8}\} &
\{\prop{3},\prop{8},\prop{10}\} &
\{\prop{3},\prop{8},\prop{13}\} &
\{\prop{6},\prop{11}\} \\[2pt]
\{\prop{8},\prop{11}\}. &&&
\end{array}
\]
Every inconsistent positive property set contains at least one of these.
These two summaries are consequences of the global closure analysis and have
an independent Lean consistency certificate; they are not additional
assumptions in the proof basis.

\begin{corollary}[Axiom bases for the full property theory]\label{cor:axiom-bases}
Let $E=\{\prop{1},\ldots,\prop{15}\}$, viewed as a theory whose fifteen
properties are axioms.  Call $B\subseteq E$ an \emph{irredundant axiom basis}
for $E$ when $B\models E$ and deleting any one axiom from $B$ destroys that
entailment.  Then the irredundant axiom bases for $E$ are exactly the seventeen
minimal inconsistent sets listed above.  Six have minimum cardinality two:
\[
\{\prop{1},\prop{2}\},\quad
\{\prop{1},\prop{6}\},\quad
\{\prop{1},\prop{8}\},\quad
\{\prop{2},\prop{11}\},\quad
\{\prop{6},\prop{11}\},\quad
\{\prop{8},\prop{11}\}.
\]
\end{corollary}
\begin{proof}
The full theory $E$ is inconsistent, already because it contains
\prop{1}+\prop{2}.  Hence $B\models E$ exactly when $B$ is inconsistent: an
inconsistent $B$ entails everything, whereas a consistent $B$ has a model and
that model cannot satisfy the inconsistent theory $E$.  Requiring each axiom of
$B$ to be essential is therefore exactly inclusion-minimal inconsistency.  The
claim follows from the seventeen-set classification above.
\end{proof}

\begin{remark}[What the consistency summaries determine]
This logical use of \emph{basis}, suggested by Jan Odelstad, should not be
confused with the \emph{Horn basis} in Section~\ref{sec:table}.  The former is
an irredundant set of property axioms entailing the full (inconsistent) theory
$E$; the latter is a generating set of sound consequence rules for computing
closures of arbitrary property sets.

Ian Hodkinson pointed out a general finite-set reason why the MCS/MIS boundary
is already enough to classify compatibility.  Let $U$ be finite, let
$\mathcal C$ be an inclusion-antichain of satisfiable subsets of $U$, and let
$\mathcal D$ consist of the inclusion-minimal subsets of $U$ that are contained
in no member of $\mathcal C$.  If every member of $\mathcal D$ is
unsatisfiable, then $\mathcal C$ is exactly the family of maximal satisfiable
subsets of $U$, and $\mathcal D$ is exactly the family of minimal
unsatisfiable subsets.  For if a satisfiable $S$ were contained in no member of
$\mathcal C$, finiteness would give a minimal subset of $S$ with that property,
hence an unsatisfiable member of $\mathcal D$, a contradiction.  The
antichain condition then gives maximality of the members of $\mathcal C$; the
dual minimality statement for $\mathcal D$ follows directly.

Applied here, this explains why M1--M6 together with the seventeen minimal
inconsistent sets determine the entire satisfiability boundary.  They do
\emph{not}, however, determine positive entailment.  For example,
\prop{3}+\prop{13} is contained only in M1, and M1 also contains \prop{7}, yet
Proposition~\ref{prop:compound-nonent}(i) shows
\[
\prop{3}+\prop{13}\catnotentails\prop{7}.
\]
The reason is that maximal positive property sets record which positive
properties can be jointly satisfied, but do not record negative literals such
as $\neg\prop{7}$.  The signed countermodel profiles used in
Section~\ref{sec:verification} carry precisely this additional information.
\end{remark}

\section{Catalogue basis and global closure status}\label{sec:table}

Table~\ref{tab:basis} gives the Horn basis used for positive-property closure
in a form that can be iterated directly.  The property-to-property part is
ordered lexicographically by antecedent after the solo rules.  To make the
degeneracy layer deductively explicit, write $\mathsf{E}$ for the assertion
$R=\varnothing$ and $\bot$ for inconsistency on the stipulated non-empty
domain.  Once $\bot$ is reached, ordinary classical consequence is explosive.

The structural rows are intentionally a \emph{basis}, not a list of every
useful degeneracy observation.  For example, \prop{1}+\prop{6} is
inconsistent and \prop{7}+\prop{8} is equivalent to $R;R=\varnothing$; these
are derivable from the displayed basis together with the propositions above.
The global Lean certificate uses an expanded 51-rule implementation in which
useful consequences concerning $R;R=\varnothing$, $\mathsf{E}$, $\bot$, and
the universal relation are made explicit.  The property-to-property part is
exactly the one displayed below; the extra structural rows are deductively
redundant consequences of this compact presentation and the propositions in
Section~\ref{sec:incompat}.

\begin{longtable}{@{}p{3.45cm}p{4.65cm}p{5.55cm}@{}}
\caption{Horn basis for arbitrary-domain closure of the fifteen properties.}
\label{tab:basis}\\
\toprule
Antecedent(s) & Consequent & Source / comment \\
\midrule
\endfirsthead
\toprule
Antecedent(s) & Consequent & Source / comment \\
\midrule
\endhead
\multicolumn{3}{@{}l}{\textit{Solo property rules}}\\
\prop{1} & \prop{3}+\prop{10} & Proposition~\ref{prop:solo} \\
\prop{4} & \prop{14} & Proposition~\ref{prop:solo} \\
\prop{6} & \prop{2}+\prop{5} & Proposition~\ref{prop:solo} \\
\prop{8} & \prop{2} & Proposition~\ref{prop:solo} \\
\prop{11} & \prop{1}+\prop{12} & Proposition~\ref{prop:solo} \\
\prop{12} & \prop{15} & Proposition~\ref{prop:solo} \\
\prop{13} & \prop{10}+\prop{14}+\prop{15} & Proposition~\ref{prop:solo} \\
\midrule
\multicolumn{3}{@{}l}{\textit{Compound property rules (lexicographic order)}}\\
\prop{1}+\prop{4}+\prop{15} & \prop{13} & Proposition~\ref{prop:p1p4p15} \\
\prop{1}+\prop{9} & \prop{11} & Proposition~\ref{prop:p1p9} \\
\prop{1}+\prop{12} & \prop{11} & Proposition~\ref{prop:p1p12p11} \\
\prop{1}+\prop{13} & \prop{4}+\prop{7} & Proposition~\ref{prop:refl-eucl} \\
\prop{1}+\prop{15} & \prop{14} & Proposition~\ref{prop:p1p15} \\
\prop{2}+\prop{5} & \prop{6} & Proposition~\ref{prop:irr-clusters} \\
\prop{2}+\prop{7} & \prop{6} & Proposition~\ref{prop:irr-clusters} \\
\prop{3}+\prop{4}+\prop{7} & \prop{1} & Proposition~\ref{prop:closure-entailments}(i) \\
\prop{3}+\prop{4}+\prop{9}+\prop{15} & \prop{12} & Proposition~\ref{prop:closure-entailments}(vii) \\
\prop{3}+\prop{7}+\prop{9} & \prop{14} & Proposition~\ref{prop:closure-entailments}(iv) \\
\prop{3}+\prop{7}+\prop{15} & \prop{14} & Proposition~\ref{prop:closure-entailments}(v) \\
\prop{3}+\prop{8}+\prop{9} & \prop{4} & Proposition~\ref{prop:closure-entailments}(ii) \\
\prop{3}+\prop{8}+\prop{15} & \prop{14} & Proposition~\ref{prop:closure-entailments}(vi) \\
\prop{4}+\prop{5} & \prop{7} & Proposition~\ref{prop:sym-antisym} \\
\prop{4}+\prop{7} & \prop{13} & Proposition~\ref{prop:sym-eucl} \\
\prop{4}+\prop{12} & \prop{9} & Proposition~\ref{prop:sym-connex} \\
\prop{4}+\prop{13} & \prop{7} & Proposition~\ref{prop:sym-eucl} \\
\prop{5}+\prop{9} & \prop{7} & Proposition~\ref{prop:p5p9} \\
\prop{5}+\prop{13} & \prop{7} & Proposition~\ref{prop:p5p13} \\
\prop{6}+\prop{9} & \prop{7} & Corollary~\ref{cor:strictweak} \\
\prop{7}+\prop{12} & \prop{9} & Proposition~\ref{prop:trans-connex} \\
\prop{8}+\prop{9}+\prop{14} & \prop{4} & Proposition~\ref{prop:closure-entailments}(iii) \\
\prop{9}+\prop{13} & \prop{7} & Proposition~\ref{prop:p9p13} \\
\prop{12}+\prop{13} & \prop{7}+\prop{9} & Proposition~\ref{prop:p12p13} \\
\midrule
\multicolumn{3}{@{}l}{\textit{Structural closure rules}}\\
\prop{1}+\prop{2} & $\bot$ & immediate \\
\prop{2}+\prop{13} & $\mathsf{E}$ & Proposition~\ref{prop:forceempty} \\
\prop{3}+\prop{7}+\prop{8} & $\bot$ & Corollary~\ref{cor:p3p7p8} \\
\prop{4}+\prop{6} & $\mathsf{E}$ & Proposition~\ref{prop:forceempty} \\
\prop{7}+\prop{8}+\prop{14} & $\mathsf{E}$ & Proposition~\ref{prop:nocomp-conv} \\
\prop{8}+\prop{10} & $\mathsf{E}$ & Proposition~\ref{prop:forceempty} \\
$\mathsf{E}$ &
$\begin{aligned}
&\prop{2}+\prop{4}+\prop{5}+\prop{6}\\
&{}+\prop{7}+\prop{8}+\prop{9}+\prop{10}\\
&{}+\prop{13}+\prop{14}+\prop{15}
\end{aligned}$ &
Proposition~\ref{prop:empty} \\
\bottomrule
\end{longtable}

\section{Conclusion}

The fifteen properties considered here do not form a simple linear hierarchy.
They organise into several overlapping mechanisms.  Diagonal information
(reflexivity or irreflexivity) interacts strongly with connexity and
Euclideanness; symmetry converts predecessor-sharing conditions into
successor-sharing ones; negative transitivity combines with antisymmetry or
Euclideanness to recover ordinary transitivity; and density is best understood
as the factorisation inequality $R\subseteq R;R$, not as strict order-density.
These mechanisms explain why apparently similar conditions can behave very
differently under composition.

The relation-algebraic translation makes that structure more transparent
without replacing the elementary first-order proofs.  At the same time, the
modal perspective explains why a pragmatic selection of relation properties is
worth studying: reflexivity, transitivity, Euclideanness, density and
convergence are not isolated textbook adjectives but frame conditions attached
to recurring modal principles.  The resulting catalogue is therefore intended
to be usable in both directions: from a list of known relational properties to
its logical closure, and from a named modal or order-theoretic class back to a
small generating set of conditions.

The global certificate provides a stronger closure guarantee than a
hand-assembled list of laws.  It checks the Horn-normalised closure of all
$32{,}768$ antecedent sets inside Lean.  Its 51 numerical rules are coupled to
semantic soundness proofs, its 49 witness rows are coupled to certified
relations and exact or residual property profiles, and every
antecedent/consequent obligation is checked by kernel reduction.  The six
maximal consistent sets and seventeen minimal inconsistent sets have an
additional independent Lean certificate.  Consequently
Theorem~\ref{thm:catalogue-complete} gives completeness for positive entailment
and compatibility among the chosen fifteen properties on non-empty domains.
The Python enumeration and Z3 searches serve as discovery and reproducibility
tools but lie outside the trust boundary for that claim.  Questions outside
this scope, like a complete catalogue of consequences peculiar to finite
domains, remain separate.

\section*{Acknowledgements}
My work on this paper began in the early 1990s.  I am grateful to my mentor and
supervisor Per-Erik Malmnäs, who also provided key pointers to the non-English
literature on the topic; to Paul Needham for several modal logic courses; and to
my fellow student Filip Widebäck for conceptual proof-theoretic advice.  My
logic teacher and later colleague Torkel Franzén provided invaluable support for
understanding what constitutes a proof.  Ian Hodkinson tirelessly explained
algebraic notions relevant to binary relations and provided valuable comments
on a late draft.  Jan Odelstad provided detailed comments from the perspectives of elementary
logic and preference theory; his question about axiom bases led to
Corollary~\ref{cor:axiom-bases}.  I take full responsibility for all errors.

\end{document}